\documentclass[runningheads]{llncs}
\usepackage[T1]{fontenc}
\usepackage{graphicx}
\usepackage{algpseudocode}
\usepackage{amsmath}
\usepackage{graphicx}
\usepackage{algorithm}

\usepackage{xcolor}

\date{December 2022}
\graphicspath{Diagrams}

\begin{document}
\title{Mitigating Misinformation Spreading in Social Networks Via Edge Blocking}

\author{Ahad N. Zehmakan\inst{1} \and
Khushvind Maurya\inst{2, 3}}
\authorrunning{A. N. Zehmakan et al.}
\institute{The Australian National University, \email{ahadn.zehmakan@anu.edu.au} \and
Indian Institute Of Technology (IIT), Delhi,
\email{khushvind.iitd@gmail.com}
\and {corresponding author}}

\maketitle
\begin{abstract}
The wide adoption of social media platforms has brought about numerous benefits for communication and information sharing. However, it has also led to the rapid spread of misinformation, causing significant harm to individuals, communities, and society at large. Consequently, there has been a growing interest in devising efficient and effective strategies to contain the spread of misinformation. One popular countermeasure is blocking edges in the underlying network.

We model the spread of misinformation using the classical Independent Cascade model and study the problem of minimizing the spread by blocking a given number of edges. We prove that this problem is computationally hard, but we propose an intuitive community-based algorithm, which aims to detect well-connected communities in the network and disconnect the inter-community edges. Our experiments on various real-world social networks demonstrate that the proposed algorithm significantly outperforms the prior methods, which mostly rely on centrality measures.

% Particularly, several source-agnostic edge-blocking methods have been developed, primarily focusing on identifying important or influential edges based on various centrality measures.

% By targeting edges with high centrality, these approaches aim to minimize the spread of misinformation in social networks. By leveraging other properties of social networks, such as communities, to determine which edges should be blocked, can potentially help to contain rumours more efficiently.

% In this paper, Blocking Edge problem is studied, and a new source ignorant method is proposed that uses community detection techniques to identify edges that connect different communities in large social networks. Further, we block these edges to contain the rumours within their communities of origin. Experimental results demonstrate that the proposed method outperforms other source-ignorant methods.
\keywords{Social Networks \and Misinformation Spreading \and Countermeasure \and Edge Blocking \and Community Detection.}
\end{abstract}

\section{Introduction} \label{introduction}

In recent years, social media has revolutionized the way individuals connect with each other, share information, and express themselves. It has created new opportunities for political engagement, social activism, and community building, and has enabled individuals to access information and resources that were previously unavailable. Social media has also transformed the way businesses and organizations interact with customers and stakeholders, providing new channels for marketing, customer service, and public relations.

However, the widespread adoption of social media platforms has undeniably resulted in a significant increase in the dissemination of misinformation. This issue permeates various domains such as politics, economics, and sociology~\cite{eismann2021diffusion}. For example, following the breach of The Associated Press Twitter account, a fabricated announcement circulated, stating that "Breaking: Two Explosions in the White House and Barack Obama is injured." As a consequence, this false information led to a staggering loss of 10 billion USD in just a few hours and triggered a rapid crash in the US stock market, cf.~\cite{peter2013bogus}.

The far-reaching consequences of misinformation spreading in different contexts cannot be understated, as they possess the potential to shape public opinion, influence decision-making processes, and even impact social cohesion. The need to address the challenges posed by the rampant spread of misinformation across diverse topics has emerged as a critical concern in today's interconnected digital age.

The spread of misinformation on social media is a complex and multifaceted phenomenon that involves a range of factors, including the structure of social networks, the psychology of misinformation diffusion, and the role of technology in shaping information sharing. Understanding these factors is critical for developing effective interventions to minimize the spread of misinformation.
% Minimizing the spread of misinformation in social networks has been a hot topic in social network analysis.

In order to mitigate the spread of misinformation in social networks, several approaches have been proposed in the past. Reducing the spread of misinformation can be achieved by some form of blocking, where a set of nodes or edges are identified and blocked from the network, under some budget constraints. Blocking a node implies that the account is either removed or banned and all its connections with other nodes are suspended, while blocking an edge implies that the connection between the two nodes connected by the edge is suspended, for example by not exposing posts from one user to another.

To effectively combat the spread of misinformation within social networks, the initial step involves promptly identifying the misinformation as it emerges, this can be done using various misinformation detection techniques, cf.~\cite{wu2019misinformation,zubiaga2018detection}.
% \cite{} discusses different methods of classifying messages into misinformation and non-misinformation. 
However, it is important to recognize that detected misinformation may resurface in modified forms, highlighting the significance of monitoring subsequent posts associated with a piece of misinformation. As a result, to minimize the spread of misinformation within a network, two popular approaches can be used:
\begin{itemize}
    \item \textbf{Source-Aware Approach} of misinformation contamination relies on identifying the sources responsible for propagating the misinformation and the users who have accepted and disseminated it within a social network. Then, containment strategies can be implemented to minimize their influence and curb the spreading ability of the misinformation.
    \item \textbf{Source-Agnostic Approach} focuses on mitigating the flow of misinformation within the network, without the prior knowledge of the specific sources of the misinformation. Through the implementation of a containment strategy, the aim is to reduce the overall dissemination of misinformation without specifically targeting its sources.
\end{itemize}

The source-aware approach is generally more powerful since it benefits from some extra source of information. However, there are two fundamental issues. Firstly, the identification of misinformation sources involves complexities such as data collection and network evolution, and source identification algorithms could be inaccurate. (Shelke and Attar~\cite{shelke2019source} provide a compressive overview of various source detection approaches and the challenges faced with these methods.) Furthermore, the source-aware approach disregards the anonymity and privacy of individuals within the social networks, up to a large extent. Consequently, there has been a growing interest, cf.~\cite{kuhlman2013blocking,zareie2022rumour}, in devising effective source-agnostic strategies that aim to minimize the flow of misinformation in the network without explicitly targeting individual sources. The present work also falls under the umbrella of this line of research.
% The focus is on identifying key network components that are most instrumental in spreading the misinformation, rather than on the misinformation's specific origin. 

% In contrast, the source-agnostic approach offers a compelling alternative that mitigates these challenges. By focusing on minimizing the flow of misinformation in the network without explicitly targeting individual sources, the source-agnostic approach strikes a balance between effective misinformation control and upholding user privacy while avoiding all fundamental challenges around misinformation detection. The focus is on identifying key network components that are most instrumental in spreading the misinformation, rather than on the misinformation's specific origin.

Two commonly employed strategies to contain the misinformation spreading are node and edge blocking. Edge blocking has garnered greater attention recently, cf.~\cite{yao2015minimizing,tong2017efficient,kuhlman2013blocking,zareie2022rumour}, since it is less intrusive (i.e., disrupts the original functionality and flow of the network less aggressively) and provides controlling power in a more granular level (note that usually blocking a node is equivalent to blocking all its adjacent edges).

In the present work, we focus on designing an effective and efficient source-agnostic edge-blocking strategy. To model the spread of misinformation, we exploit the popular Independent Cascade model~\cite{kempe2003maximizing}. We investigate the problem of minimizing the expected number of nodes that will be exposed to a piece of misinformation when we are allowed to block $k$ edges for some given integer $k$. We show that this problem is NP-hard. (It is worth stressing that while this problem has been extensively studied by prior work and several approximation approaches were proposed~\cite{zareie2022rumour}, we are the first to formally prove that the problem is computationally hard.)

We propose an intuitive community-based algorithm, which first uses a community detection algorithm such as Louvain community detection algorithm~\cite{blondel2008fast}, to partition the nodes into communities (i.e., subsets of well-connected nodes). Then, we try to slow down the flow of misinformation between these communities by disconnecting the inter-community edges. The idea is that stopping the spread of misinformation inside a community is hard since it requires blocking a significant number of edges. However, there are substantially less edges between communities whose blocking could drastically reduce the extent that the misinformation travels.

We provide our experimental findings on several real-world graph data. We observe that our proposed algorithm consistently and significantly outperforms the existing algorithms.

% Source-ignorant methods aim to identify critical nodes or edges in a network that, when blocked, can minimize the flow of information in response to the detection of a misinformation, without specifying the misinformation's original source.
% In this paper we propose a source-ignorant edge-blocking method, that uses the Louvain community detection algorithm\cite{blondel2008fast} to identify and block edges that join different communities, thus restricting the spread of misinformation within the communities of sources the misinformation.  This new algorithm is the main contribution of this paper and consists of two different steps. First, Different communities are identified using Louvain Community detection algorithm. Second, Edges which join these different communities are identified. Third step involves blocking these edges, thus disconnecting all the communities.

\noindent \textbf{Outline.} The rest of the paper is structured as follows. Section~\ref{preliminaries} covers the preliminaries. Then, in Section~\ref{prior_work} we give an overview of related previous work. The hardness results are provided in Section~\ref{Hardness}. Our proposed algorithm is presented in Section~\ref{algorithm}. Finally, our experimental findings and comparison of algorithms are provided in Section~\ref{evaluation}.

\subsection{Preliminaries} \label{preliminaries}

\subsubsection{Graph Definitions.} Let $G=\left(V,E, \omega \right)$ be a weighted graph, where function $\omega:E\rightarrow [0,1]$ assigns a value between $0$ and $1$ to each edge in the graph. Let us define $n:=|V|$ and $m:=|E|$. 
For a node $v\in V$, $N\left(v\right):=\{v'\in V: (v,v') \in E\}$ is the \emph{neighborhood} of $v$. Furthermore, $\hat{N}(v):=N(v)\cup \{v\}$ is the \textit{closed neighborhood} of $v$.
Let $d\left(v\right):=|N\left(v\right)|$ be the \emph{degree} of $v$ in $G$.
% We also define $d_A(v):=|N(v)\cap A|$ for a set $A\subseteq V$.
The \textit{girth} of a graph $G$ is the length of the shortest cycle contained in the graph.
If $G$ has no cycle, then the girth is defined to be infinity.

\subsubsection{Independent Cascade Model~\cite{kempe2003maximizing,goldenberg2001talk}.} Each node can have one of the following three states:
\begin{itemize}
    \item \textit{Ignorant (white)}: A node which has not heard of the misinformation.
    \item \textit{Spreader (red)}: A spreader is a node who has heard the misinformation and spreads it.
    \item \textit{Stifler (orange)}: A node who has heard the misinformation but does not spread it.
\end{itemize}

Let a coloring $\mathcal{C}$ be a function $\mathcal{C}:V\rightarrow \{w,r,o\}$, where $w$, $r$, and $o$ correspond to white, red, and orange, respectively. The process starts from an initial coloring $\mathcal{C}_0$. Then, in each round $t\in N$, all nodes simultaneously update their state according to following updating rules:
\begin{itemize}
    \item A white node $v$ becomes red with probability:
\begin{equation}
    \label{p^*-eq}
    p^*(v):= 1-\prod_{v'\in N(v)\& \mathcal{C}_{t-1}(v')=r}\left(1-\omega\left((v,v')\right)\right).
\end{equation}
    \item A red node becomes orange.
    \item An orange node remains orange.
\end{itemize}
More precisely, we have:

$\mathcal{C}_t(v)$ =
$\begin{cases}
r & \text{if } \mathcal{C}_{t-1}(v)=w \text{ with probability } p^*(v)\\ 
w & \text{if } \mathcal{C}_{t-1}(v)=w \text{ with probability } 1-p^*(v)\\
o & \text{if } \mathcal{C}_{t-1}(v)=o \lor \mathcal{C}_{t-1}(v)=r.
\end{cases}$

Let $(v,v')\in E$, where $v$ is white and $v'$ is red. Then, $v'$ makes $v$ red with probability $\omega((v,v'))$. This explains the choice of probability $p^*(v)$ in Equation~\eqref{p^*-eq}. Furthermore, a red node has one chance to spread, and then it becomes orange (stifler) and remains orange forever. This model is usually known as the \textit{Independent Cascade} (IC) model, cf.~\cite{goldenberg2001talk,kempe2003maximizing}.  \ \\

The main focus of the present paper is to devise an effective edge-blocking strategy to minimize the spread of misinformation, simulated by the Independent Cascade model. An exact formulation of the problem is given in Section~\ref{Hardness}.

\subsection{Prior Work}\label{prior_work}
% A wide range of rumour spreading models have been developed and studied \cite{n2020rumor,zehmakan2023random}
% In this section, we first give an overview of some well-established (mis)-information spreading models and further present the related previous research on countermeasures.

\subsubsection{Information Spreading Models.} 
\label{prio-info}
A plethora of (mis)-information spreading models have been developed and studied in recent years, cf.~\cite{kempe2003maximizing,zehmakan2023rumors,gartner2020threshold,zehmakan2019spread,zehmakan2023random,zehmakan2021majority}. Here, we focus on the most fundamental and relevant models.
\begin{itemize}
    \item \textbf{Independent Cascade Model~\cite{goldenberg2001talk,kempe2003maximizing}.}
As described in Section~\ref{preliminaries}, this is a model of information diffusion that assumes that information spreads through a network of individuals in a series of steps. In every step, it considers each node to be in one of the three states, red (spreader), white (ignorant) and orange (stifler). Then, each red node gets one chance (before becoming orange) to make its white neighbors red (i.e., inform them). Motivated by viral marketing, the main focus in this model is to develop algorithms for finding subsets of nodes that maximize the spread of the red color, mostly exploiting monotonicity and submodularity properties~\cite{mossel2007submodularity}. Different variants of the Independent Cascade model have also been studied, for example, where there is a forgetting mechanism in place~\cite{zehmakan2023rumors} or when there are more than one pieces of (mis)-information spreading~\cite{liu2019influence}
    % This model is a model of information diffusion that assumes that information spreads through a network of individuals in a series of steps. in every step, It considers each node to be in one of the three states, red (node that has heard the misinformation and spreads it), white (node which has not heard of the misinformation) and orange (node that has heard the misinformation but does not spread it). An individual  with one or more infected neighbours is susceptible to getting infected. In each round, a white node $v$ becomes red w.p. $p^*(v)$ as defined in Section \ref{preliminaries}. 
\end{itemize}
\begin{itemize}
    \item \textbf{Linear Threshold Model~\cite{kempe2003maximizing}.} This model of information diffusion assumes that individuals are more likely to be exposed to some (mis)-information if a larger fraction of their neighbors have been exposed to it. More precisely, each node $v$ has a threshold $\tau_v$ chosen randomly from the interval $[0,1]$. Then, a white node $v$ becomes red once at least $\tau_v$ fraction of its neighbors are red, and then remains red forever. In the Linear Threshold model also, motivated by viral marketing, the problem of finding a set of seed nodes which can maximize the final number of red nodes has been extensively studied, cf.~\cite{kempe2003maximizing,rahimkhani2015fast}. Generalized variants of threshold model have been considered too, cf.~\cite{kothari1985generalized}.
% \begin{equation}
%     \sum_{u \ active \ neighbour \ of \ v} w(u,v) > \theta_v
% \end{equation}
% \end{itemize}
% \begin{itemize}
%     \item \textbf{Stochastic model:} In a stochastic model, the state of each individual is determined by a random process. For example, if an individual has a $50\%$ chance of becoming infected when it is exposed to the disease, then it will become infected with a $50\%$ probability. 
% \end{itemize}
% \begin{itemize}
%     \item \textbf{Deterministic Model:} In a deterministic model, the state of each individual is determined by the states of its neighbors. For example, if all of an individual's neighbors are infected, then the individual will become infected.
% \end{itemize}
% \begin{itemize}
    \item \textbf{Susceptible-Infected-Recovered (SIR) Model~\cite{daley1964epidemics}.} The SIR model is a commonly used epidemiological model that describes the spread of infectious diseases in a population. It divides the population into three compartments: Susceptible, Infected, and Recovered. Then, each node can change its state following a predefined stochastic updating rule relying on some model parameters, infection and recovery factor. While the model was originally introduced to emulate the spread of diseases, it has gained some popularity in modeling (mis)-information spreading. The original model assumes that the homogeneous mixing condition holds, that is, the nodes are connected via a clique. However, network based variants of the model have been studied as well, cf.~\cite{kenah2007network}.
    
    % The model assumes that once individuals recover from the disease, they become immune and cannot be infected again. Epidemiological models define a parameter $\beta$ that combines consequences of the rate of social contact $\chi$ and the rate of transmission upon contact, $\tau$, namely $\beta=\tau\chi$. The infected population ($I$) changes because every infected person recovers with a probability of $\gamma$ per period.
    
%     The changes in the population in different compartments are given by
% \begin{equation}
%     \Delta S_{t+1} =-\beta I_t (S_t /N),
% \end{equation}
% \begin{equation}
%     \Delta I_{t+1} =-\beta (S_t/N)I_t -\gamma I_t,
% \end{equation}
% \begin{equation}
%     \Delta R_{t+1} =-\gamma I_t,
% \end{equation}

\end{itemize}

\subsubsection{Countermeasures.}
\label{prio-counter}

Reducing the propagation of misinformation is a significant challenge in the field of social network analysis and has garnered considerable interest. Various approaches have been proposed to tackle this issue and mitigate the spread of false information in social networks.

\paragraph{Edge Blocking Countermeasure:}
 One countermeasure which has gained significant popularity is edge blocking, cf.~\cite{yao2015minimizing,tong2017efficient,kuhlman2013blocking,zareie2022rumour,n2020rumor,zehmakan2023rumors}.
Holme et al.~\cite{holme2002attack} considered four different edge blocking strategies: blocking by the descending order of the degree and the betweenness centrality, calculated for either the initial network or the resulting network during the blocking procedure. It is observed edges blocked in order of betweenness show more efficient misinformation mitigation as compared to edges blocked in decreasing order of degree. Kimura et al.~\cite{kimura2007extracting} introduced a method of efficiently estimating the influence of nodes using bond percolation. This bond percolation method then was used in~\cite{kimura2008minimizing,kimura2009blocking} to identify a set of edges which, when blocked, maximize the contamination degree of the network. Yan et al.~\cite{yan2019rumor} proposed a greedy method to identify the most critical edges among a set of candidate edges to minimize the spread of a misinformation. Pagerank centrality~\cite{brin1998anatomy} is used in~\cite{yan2019rumor} as a criterion for blocking the edges to minimize the spread of misinformation. The susceptibility of a graph to diffusion is defined in~\cite{khalil2013cuttingedge} as the sum of the expected influence of each node when it is the single source for a cascade. Further, a greedy method is proposed that minimizes the spread susceptibility of the network. Tong et al.~\cite{tong2012gelling} provided an approach where the edges blocked depend on the eigenvalue of the adjacency matrix of the network. Finally, in a very recent work, Zareie and Sakellariou~\cite{zareie2022rumour} took into account additional features of edges (beyond centrality), such as entropy, to determine what edges to block. Some more results on edge blocking problem are discussed in, \cite{burzyn2006np,yannakakis1978node}.

\paragraph{Other Countermeasures:}
Motivated by blocking accounts in real-world online social platforms, the countermeasure of blocking nodes has been widely studied. Various node blocking methods have been investigated in the literature, that use degree centrality, betweenness centrality and closeness centrality as a criterion to block nodes, cf.~\cite{he2015modeling,wang2015maximizing,habiba2010finding,dey2017centrality}. In~\cite{taninmics2020minimizing}, the authors proposed two heuristic algorithms for minimizing the spread of misinformation simulated by the Independent Cascade model via node blocking. Pham et al.~\cite{pham2018maximizing} studied a variant of the problem with some time and budget constraints. Schneider et al.~\cite{schneider2011suppressing} considered the setup where the sum of the sizes of the connected large clusters in the network is considered as an information flow metric and nodes with high betweenness centrality are suggested to be blocked to minimize the sum of the sizes.

Some other countermeasures have also been considered. For example, the authors of~\cite{tripathy2010study,ding2020efficient} studied truth spreading as a misinformation mitigation method, where truth is spread as anti-misinformation. Zehmakan et al.~\cite{zehmakan2023rumors} introduced a similar countermeasure, where a subset of nodes is selected to be ``fact-checkers'' whose role is to trigger the spread of truth once exposed to a piece of misinformation. See~\cite{gausen2021can,zehmakan2023rumors} for more examples of countermeasures.

% . For more results on other countermeasures, see \cite{kempe2003maximizing,zehmakan2023random,zehmakan2023rumors}
% \cite{wang2013negative} proposed a source-aware algorithm, where goal is to minimize the size of ultimately contaminated users by discovering and blocking k uninfected users.

\section{Problem Formulation and Hardness Result} \label{Hardness}
In this section, we aim to show that the problem of minimizing the spread of misinformation via blocking edges is computationally hard. It is worth emphasizing that while this problem has been studied extensively by prior work and several approximation algorithms have been put forward, cf.~\cite{yao2015minimizing,tong2017efficient,kuhlman2013blocking,zareie2022rumour}, this work is the first to analyze the computational complexity of this problem rigorously. Let us first provide a more concrete formulation of the problem.\\

\noindent\textsc{Edge Blocking Problem.}\\
\noindent\textit{\textbf{Input}}: A weighted graph $G=(V_G,E_G,\omega)$, an integer $k$, a random distribution over all possible colorings. \\
\textit{\textbf{Output}}: The maximum expected final number of white nodes when $k$ edges are blocked (i.e., removed), starting from an initial coloring $\mathcal{C}_0$ chosen from the given distribution and following the Independent Cascade model.

Our hardness result is provided in Theorem~\ref{hardness-thm}, building on the \textsc{Densest Subgraph Problem}~\cite{manurangsi2017almost}. We first need to provide some basic definitions and lemmas. \\

% \begin{definition}
% We say that $\mathcal{A}$ is a $\rho$-approximation algorithm for a maximization problem P and some $\rho>0$ if the output of $\mathcal{A}$ is not smaller than the optimal solution divided by $\rho$ for any instance of $P$.
% \end{definition}

\noindent\textsc{Densest Subgraph Problem}
\\
\textbf{Input}: A connected undirected graph $H=(V_H,E_H)$ and an integer $k<|V_H|$. \\
\textbf{Output}: The maximum number of edges in a subgraph induced by $k$ nodes in $H$. \\

\begin{theorem}[\cite{manurangsi2017almost}]
\label{dense-hardness}
The \textsc{Densest Subgraph Problem} is NP-hard.
% There is no polynomial-time $n_H^{1/(\log\log n_H)^c}$-approximation algorithm for the \textsc{Densest Subgraph Problem} (where $c>1$ is a constant) unless Exponential Time Hypothesis (ETH) does not hold.
\end{theorem}

% The Exponential Time Hypothesis states that satisfiability of 3-CNF Boolean formulas cannot be solved more quickly than exponential time in the worst case, cf.~\cite{impagliazzo2001complexity}. Similar to P$\ne$NP, The hypothesis, is widely believed to be true.\\

\textbf{Remark.} Note that to be precise, when talking about NP-hardness, we need to refer to the decision variant of the problem, where an integer $a$ is also given as input and the problem is to determine whether there is a subgraph induced by $k$ nodes which has $a$ edges.

% \noindent\textsc{Blocking Nodes (BN) Problem}
% \\
% \textbf{Input}: A weighted graph $G=(V_G,E_G,\omega)$, an initial coloring $\mathcal{C}_0$, and an integer $k$. \\
% \textbf{Output}: The maximum final number of white nodes (in expectation) when $k$ nodes are blocked. \\

\begin{lemma}
\label{special-case-lemma}
The \textsc{Densest Subgraph Problem} is polynomial-time solvable if $k$ is smaller than the girth of $H$.
\end{lemma}
\begin{proof}
Let $OPT_{DS}(H,k)$ be the optimal solution for the \textsc{Densest Subgraph Problem} for input $H$ and $k$. Consider an arbitrary set of nodes of size $k$. The induced subgraph by this set contains at most $k-1$ nodes. This is true because otherwise, the graph has a cycle of length $k$ or smaller which is in contradiction with the assumption of the theorem. Therefore, we have $OPT_{DS}(H,k)\le k-1$. On the other hand, any set of $k$ nodes which induces a connected subgraph in $k$ has at least $k-1$ edges, which implies $OPT_{DS}(H,k)\ge k-1$. Hence, we conclude that $OPT_{DS}(H,k)= k-1$. We can check whether the condition of the lemma is satisfied in polynomial time and return $k-1$ as the answer. \qed
\end{proof}

% \noindent\textsc{Blocking Edges (BE) Problem}
% \\
% \textbf{Input}: A weighted graph $G=(V_G,E_G,\omega)$, an initial coloring $\mathcal{C}_0$, and an integer $k$. \\
% \textbf{Output}: The maximum final number of white nodes (in expectation) when $k$ edges are blocked. \\

\begin{definition}[Transformer]
\label{transformer2}
Consider a connected undirected graph $H=(V_H,E_H)$, where $V_H:=\{v_1,\cdots, v_{n_H}\}$ and $E_H:=\{e_1,\cdots, e_{m_H}\}$. Construct graph $G=(V_G, E_G,\omega)$ as follows:
\begin{itemize}
    \item $V_G:=X\cup Y\cup \{z\}$ where $X:=\{x_1,\cdots,x_{n_H}\}$ and $Y:=\{y_1,\cdots, y_{m_H}\}$.
    \item $E_G:=\{(y_j,x_i):v_i\in e_j\}\cup \{(x_i,z):1\le i\le n_H\}$.
    \item $\omega(e)=1$ for $e\in E_G$.
\end{itemize}
% Note that $n_G:=|V_G|=n_H+m_H+1$ and $m_G:=|E_G|=2m_H+n_H$.
% See Fig.~\ref{figure-transformer} for an example.
\end{definition}

% \begin{definition}[Transformer]
% \label{transformer1}
% Consider a graph $H=(V_H,E_H)$, where $V_H:=\{v_1,\cdots, v_{n_H}\}$ and $E_H:=\{e_1,\cdots, e_{m_H}\}$. Construct graph $G=(V_G, E_G)$ as follows:
% \begin{itemize}
%     \item $V_G:=X\cup Y$, where $X:=\{x_1,\cdots,x_{n_H}\}$ and $Y:= \{y_1,\cdots, y_{m_H}\}$.
%     \item $E_G:=\{\{x_i,y_j\}:v_i\in e_j\}$.
%     \item $\omega(e)=1$ for $e\in E_G$.
% \end{itemize}
% Note that $n_G:=|V_G|=n_H+m_H$ and $m_G:=|E_G|=2m_H$. (See Figure~\ref{figure-transformer} for an example.)
% \end{definition}

% \vspace{-1cm}
\begin{figure}[h]
  \centering
  \includegraphics[width=0.5\linewidth]{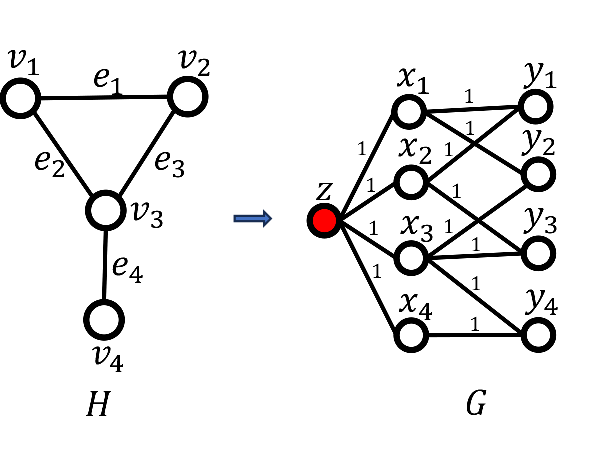}
  \caption{An example graph $H$ and the obtained graph $G$ after applying Transformer (from Definition~\ref{transformer2}).}
  \label{figure-transformer}
\end{figure}
% \begin{theorem}
% \label{hardness-thm}
% There is no polynomial-time $n_G^{1/2(\log \log n_G)^{c'}}$-approximation algorithm for the \textsc{Edge Blocking Problem} for some $c'>1$, unless ETH does not hold.
% \end{theorem}

\begin{theorem}
\label{hardness-thm}
The \textsc{Edge Blocking Problem} is NP-hard, even when all edges have weight 1.
\end{theorem}
\begin{proof}
The proof builds on a reduction from the \textsc{Densest Subgraph Problem}. Assume that there is a polynomial-time algorithm $\mathcal{A}$ for the \textsc{Edge Blocking Problem} (when all edge weights are 1). Let $H=(V_H,E_H)$, for $V_H:=\{v_1,\cdots, v_{n_H}\}$ and $E_H:=\{e_1,\cdots, e_{m_H}\}$, and $k$ be the input of the \textsc{Densest Subgraph Problem}. If $k$ is smaller than the girth of $H$, then we can solve the problem in polynomial time according to Lemma~\ref{special-case-lemma}. Otherwise, we use Transformer from Definition~\ref{transformer2} to build graph $G$ from $H$. Furthermore, consider the coloring $\mathcal{C}_0$ where node $z$ is colored red and the rest of nodes are white, see Fig.~\ref{figure-transformer} for an example. We define the random distribution to pick this coloring with probability 1.
Let $OPT_{DS}(H,k)$ be the optimal solution to the \textsc{Densest Subgraph Problem} for the input $H$ and $k$ and $OPT_{EB}(G,k,\mathcal{C}_0)$ be the optimal solution of the \textsc{Edge Blocking Problem} for the input $G$, $k$, and $\mathcal{C}_0$. We prove that
\begin{equation}
\label{eq:opt}
OPT_{DS}(H,k)=OPT_{EB}(G,k,\mathcal{C}_0)-k.
\end{equation}

Note that the Transformer generates $G$ from $H$ in polynomial time. Thus, there is an algorithm $\mathcal{A'}$ for the \textsc{Densest Subgraph Problem} which first executes Transformer. Then, it runs the algorithm $\mathcal{A}$ to compute $OPT_{EB}(G,k,\mathcal{C}_0)$ and subtracts it by $k$ to obtain $OPT_{DS}(H,k)$ (using Equation~(\ref{eq:opt})). This algorithm clearly runs in polynomial time. This implies that the \textsc{Edge Blocking Problem} is NP-hard based on Theorem~\ref{dense-hardness}.

It remains to prove that Equation~(\ref{eq:opt}) holds. Let a set $S\subset V_H$ of size $k$ induce a subgraph with $OPT_{DS}(H,k)$ edges. Define $X_S:=\{x_i:v_i\in S\}$, which is of size $k$. If we block the edges between nodes in $X_S$ and $z$ (i.e., $\{(x_i,z):x_i\in X_S\}$), then all nodes in $Y_S:=\{y_j: e_j=\{v_{i_1},v_{i_2}\}\ \textrm{for}\ v_{i_1},v_{i_2}\in S\}$ remain white because they become disconnected from node $z$ (the only node which is red in $\mathcal{C}_0$). Since $|Y_S|=OPT_{DS}(H,k)$, we have $OPT_{DS}(H,k)+k\le OPT_{EB}(G,k,\mathcal{C}_0)$. (We added $k$ since nodes in $X_S$ also remain white.)

Now, we prove that $OPT_{DS}(H,k)\ge OPT_{EB}(G,k,\mathcal{C}_0)-k$. Since in $G$ all edges have weight 1, all nodes, except the ones which cannot reach $z$, become red and then orange after at most three rounds. Let $E_{XZ}$ be the set of edges from nodes in $X$ to $z$ and $E_{YX}$ be the set of edges from $Y$ to $X$. We claim that there is an optimal solution which only blocks edges in $E_{XZ}$. Let set $S_1$ be an optimal solution (i.e., blocking edges in $S_1$ makes $OPT_{EB}(G,k,\mathcal{C}_0)$ nodes remain white forever) and $e_j\in S_1$ for some $e_j \in E_{YX}$.

Note that $k$ is at least as large as the girth of $H$ (since we already excluded the other case). This implies that $OPT_{DS}(H,k)\ge k$ since a connected subgraph including the smallest cycle induces at least $k$ edges. Since we already proved that $OPT_{DS}(H,k)\le OPT_{EB}(G,k,\mathcal{C}_0)-k$, we get $2k\le OPT_{EB}(G,k,\mathcal{C}_0)$. Then, there exists a maximal subset $D\subset S_1\cap E_{XZ}$ which covers at least $|D|$ nodes in $Y$. We say a node $y_j$ is \textit{covered} if both its neighbors in $X$ are disconnected from $z$ by blocking edges in $D$. This is true because otherwise $S_1\cap E_{XZ}$ covers at most $|S_1\cap E_{XZ}|-1$ nodes in $Y$. In addition to these nodes, the only nodes which could remain white are the nodes in $\{w:(w,w')\in S_1 \textrm{ for some } w'\}$ whose size is trivially at most $|S_1|$. Thus, at most $|S_1|+|S_1\cap E_{XZ}|-1\le 2|S_1|=2k-1$ nodes remain white, which is a contradiction since we argued that the optimal solution is at least $2k$. This implies that there is such a node set $D$.

Define $D':=\{v_i:(x_i,z)\in D\}$ and $S_1^{\prime}:=\{v_i: (x_i,z)\in S_1\}$. Let $v$ be a node in $V_H\setminus S_1^{\prime}$, which has an edge $e$ into $D'$. (Such a node must exist since $k<n$ and we defined $D$ to be maximal.) Assume that $x$ is the node corresponding to $v$ in $X$ and $y$ is the node corresponding to $e$ in $Y$. If an edge from $y$ to $X$ is in $S_1$, remove it, otherwise remove another edge in $S_1\cap E_{YX}$ which must exist by assumption, and instead add $(x,z)$ to obtain $S_2$. Since node $y$ will be covered (and will remain white), the solution of $S_2$ is at least as large as the one from $S_1$. Thus, there exists an optimal solution $S$ of size $k$ such that $S\cap E_{YX}=\emptyset$. This means that there are $OPT_{EB}(G,k,\mathcal{C}_0)-k$ nodes in $Y$ which remain white (i.e., the edges between $z$ and both their neighbors are in $S$). Note that we subtracted by $k$ since the $k$ nodes in $X$ whose edge to $z$ is removed remain white too.

Let us define $S_H:=\{v_i:(x_i,z)\in S\}$. Then, $S_H$ induces a subgraph with $OPT_{EB}(G,k,\mathcal{C}_0)-k$ edges. This implies that $OPT_{DS}(H,k)\ge OPT_{EB}(G,k,\mathcal{C}_0)-k$. \qed

\end{proof}

\textbf{Remark.} Note that when all edges have weight 1, the problem of finding the final expected number of orange nodes in the IC model for a given graph $G$ and coloring $\mathcal{C}_0$ is equivalent to a reachability problem, which can be solved in polynomial time, while when any edge weight is allowed, the problem is known to be \#P-hard~\cite{kempe2003maximizing}. We proved that the \textsc{Edge Blocking Problem} is NP-hard even when all edge weights are 1. Thus, the hardness comes from the choice of $k$ edges rather than the IC model.

% This indicates the hardness is inherent to the \textsc{Edge Blocking Problem} rather than computing the expected final number of orange nodes in the Independent Cascade model, which is known to be \#P-hard~\cite{}. This is because when all edges have weight 1, finding the final number of orange nodes is equivalent to the reachability problem which is polynomial-time solvable.

% \subsection{Submodularity}
% We observe that we cannot rely on greedy approaches using submodularity, as it is done in prior work for example in~\cite{kempe2003maximizing}. Even if we prove the function we are trying to maximize is submodular, it cannot be computed efficiency since all proofs use edge weights of 1.
% However, we can still use the greedy approach from~\cite{kempe2003maximizing} as a heuristic for blocking nodes and selecting verifiers.

\section{Proposed Algorithm} \label{algorithm}

A community refers to a subset of nodes within a graph that exhibits a higher degree of interconnectedness than the rest of the network. Communities are often characterized by a greater density of edges between nodes within the community compared to edges connecting nodes between different communities.
% A community can be thought of as a cluster or subgraph of the larger network, where the nodes within the community are more strongly connected to each other than they are to nodes outside the community.
Identifying communities within social networks can provide valuable insights into the structure and dynamics of the network.
% , as well as the behavior of its individual users and trends in information-sharing among them.

There are several community detection algorithms in the literature. Some of the most popular ones are the Louvain~\cite{blondel2008fast}, Leiden~\cite{traag2019louvain}, Surprise~\cite{aldecoa2013surprise}, and Walktrap~\cite{clauset2004finding} algorithm. We rely on Louvain algorithm for community detection, which works by iteratively optimizing the modularity of a network, which is a measure of how well the nodes in a network are grouped into communities. The Louvain algorithm is fast and scalable, and it has been shown to be effective in detecting communities in a variety of networks.
% It has a resolution parameter that allows obtaining communities of different sizes.
Our algorithm uses this algorithm to first find a set of communities such that the number of inter-community edges is at most $k$, the budget for the number of edges to be blocked. Then, we simply block all these edges.

The Louvain algorithm receives a graph $G$ and a resolution parameter $r$. The value of $r$ controls the number of communities (and consequently, the number of inter-community edges) the algorithm will output. Our goal is to generate a set of communities such that the number of inter-community edges is smaller than $k$ but as close as possible to it. 

To achieve this, we employ a multi-step process, which is described in Algorithm~1. This essentially follows a hit-and-trial process by updating the resolution parameter and re-running the Louvain algorithm. In addition to graph $G$ and budget $k$, it also receives an initial resolution parameter $r$, two repetition parameters $h_1$ and $h_2$, and an increasing factor $f>1$. It initially sets $S=\emptyset$ and $count=0$. Then, it runs in a \textbf{while} loop until $count$ is larger than the number of repetitions $h_1$. Inside this, it first runs a \textbf{for} loop for $h_2$ times. Each time, it runs the Louvain algorithm and finds the inter-community edges. Then, for each of these edge sets $\mathcal{E}$, if the size of $\mathcal{E}$ is smaller than $k$, but larger than current $S$, then we update $S=\mathcal{E}$. This way, the size of $S$ gets closer to the budget $k$, but it does not exceed it. Note that we run the \textbf{for} loop $h_2$ times, since the Louvain algorithm is nondeterministic. Once the \textbf{for} loop is over, we update the resolution factor to $r=r*f$, where $f$ is the increasing factor. Furthermore, if $|\mathcal{E}|>k$, we increment $count$. Note that at the beginning, $count$ might remain zero until $r$ is large enough such that $\mathcal{E}$ (the number of inter-community edges) becomes large. Then, $count$ will increase until it exceeds $h_1$ and then the \textbf{while} loop is over. We then return the set $S$. 

% First, We select a value of resolution parameter which gives $\mid S\mid$ close to $s$, while satisfying $\mid S\mid < s$. This resolution parameter needs to be found using hit-and-trial for each dataset. We then run the community detection algorithm five times for this value of the resolution parameter. This allows us to generate multiple sets of potential $S$. Next, we iterate through slightly different values of the resolution parameter  to obtain a range of possible sets $S$. After each iteration we increase the resolution parameter by a factor 'increasing factor' ($k$). For our experiments, we have used $k = 1.01$. Finally, from this range of sets $S$, we select the one that has the closest size, denoted as $\mid S\mid$, to the desired number of edges to be removed, $s$. By following this approach, we can effectively identify and determine the optimal set $S$ for removing the desired number of edges in our algorithm.

% We can obtain different $S$ by varying the resolution parameter using some increasing factor. In our algorithm, we need to do community detection multiple times to obtain the best $\mid S \mid$. For this purpose, we have used a multi-step process, which is described in more detail in Section \ref{Modle Parameters}

\begin{algorithm}[h]
    \caption{Pseudocode for our proposed algorithm}
    \hspace*{\algorithmicindent} \textbf{Input:} $G(V,E,\omega)$, Resolution $r$, Increasing Factor $f$, Repetitions $h_1$ and $h_2$, and Budget $k$\\
    \hspace*{\algorithmicindent} \textbf{Output:} Set of edges $S$ of size at most $k$ to be blocked.
    \begin{algorithmic}[1]
        \Procedure{Algorithm}{$G,r,f,h_1,h_2,k$}
        \State S = $\emptyset$
        % \State C = Louvain Community Detection (G,r)  \Comment{'communities' is a set of disjoint communities, where each community is represented by a set of nodes}
        % \State Edges = set of all the edges $(u,v)$ such that u and v lie in different communities in C
        \State count = 0
        \While{$(count <= h_1)$}
            
            \For{\texttt{ i from 1 to $h_2$}}
                \State C = set of communities returned by the Louvain algorithm for $G$, $r$
                \State $\mathcal{E}$ = set of inter-community edges for $C$
                \If{$|\mathcal{E}| > |S|$ and $|\mathcal{E}| <= k$}
                    \State $S=\mathcal{E}$
                \EndIf
            \EndFor
            \State update $r = r*f$
            \If{$|\mathcal{E}| > k$}
                \State count++
            \EndIf
        \EndWhile
        \State \textbf{return} $S$
        \EndProcedure
    \end{algorithmic}
\end{algorithm}

% \subsection{Previous Algorithms}
% \ahad{We should give an exact description of each of these algorithms.}

% \begin{itemize}
%     \item \textbf{IEED(Iterative Efficient Edge Detection)}\cite{zareie2022rumour}:  Critical edges are identified and blocked in the network
%     \item \textbf{Highest Degree\cite{kempe2003maximizing}:} The sum of the degree values of the nodes connected to an edge is used to select critical edges.
%     \item  \textbf{Betweenness Centrality} \cite{dey2017centrality} : Edges with largest betweenness centrality are identified and blocked. 
%     \item \textbf{Pagerank Centrality} \cite{brin1998anatomy}: The sum of the pagerank of the nodes connected to each edge can determine the importance of the edge.
%     \item \textbf{Bond Percolation Method\cite{kimura2007extracting}: } In \cite{kimura2009blocking} ,the contamination degree of the network is characterized by the average and maximum influence of all nodes. The objective is to identify a set of edges that, when blocked, minimizes the contamination degree of the network. To address this objective, an approximate algorithm is proposed that leverages the bond percolation method.
%  \item \end{itemize}

% \vspace{-1cm}
\section{Evaluation} \label{evaluation}

\subsection{Experimental Setup}
\label{setup:sec}

\subsubsection{Social Networks.} For our experiments, we use three subgraphs of Facebook, namely Facebook from SNAP dataset~\cite{snapnets} and Facebook-Politician and Facebook-Govt from Network Repository~\cite{nr}. Some graph properties of these networks are listed in the table below.

% Three Datasets have been used for the experiments. Facebook~\cite{snapnets} is a dataset consisting of 'circles' (or 'friends lists') from Facebook. We have used the graph obtained after combining edges from all egonets. The FB-Politician\cite{nr} dataset comprises data collected in November 2017, The dataset represents blue verified Facebook page networks of different categories. Nodes represent the pages and edges are mutual likes among them. FB-Govt\cite{nr} is a dataset of mutually liked facebook pages, with nodes representing the pages and edges being mutual likes among them. The selection of the three datasets was based on their possession of large clustering coefficients, which indicates a higher degree of interconnectedness within the networks. Additionally, these datasets were chosen specifically for their ability to yield numerous communities when analyzed using the community detection algorithm.

\begin{table}[]
\label{stats:table}
    \centering
    \caption{Some characteristics of the networks used in the experiments, including average degree $d_{avg}$, maximum degree $d_{\max}$, diameter $D$, average clustering coefficient $K_{avg}$, and number of triangles $T$.}
    \label{characterstics of datasets}
    \begin{tabular}{|c| c| c| c| c| c| c| c|}
        \hline
           Network & $n$ & $m$ & $d_{avg}$ & $d_{max}$ & $D$ & $K_{avg}$  & $T$\\ \hline
        Facebook  &  4039 & 88234 & 43.691 & 1045 & 8 & 	0.6055 & 1612010 \\
        Facebook-Govt  &  7057 & 89429 & 25.344 & 697 & 10 & 	0.410 & 523854 \\
        Facebook-Politician  &  5908 & 41706 & 14.118 & 323 & 14 & 	0.6055 & 174632 \\  \hline
    \end{tabular}
\end{table}

\subsubsection{Edge Weights.} Most real-world networks are unweighted, and one needs to introduce a meaningful procedure for weight assignment. Using the communication information of individuals on various real-world networks, the authors in~\cite{onnela2007structure,goyal2010learning} observed that there is a strong correlation between the number of shared friends of two individuals and their level of communication. Consequently, they proposed the usage of similarity measures, such as Jaccard-like parameters, to approximate the weights of connections between nodes. This is also aligned with the well-studied strength of weak ties hypothesis~\cite{granovetter1973strength}. Therefore, we assign the edge weights according to the Jaccard index~\cite{jaccard1901etude} in our set-up. More precisely, for each edge $(v,u)\in E$, we set $\omega\left((v,u)\right)=\frac{|\hat{N}(v)\cap \hat{N}(u)|}{N(v)\cup N(v)}$.
% \begin{equation}\label{eq-jaccard}
% .
% \end{equation}
% Where $\hat{N}(v):=N(v)\cup \{v\}$ is the \textit{closed neighborhood} of $v$.
We use $|\hat{N}(v)\cap \hat{N}(u)|$ instead of $|N(v)\cap N(u)|$ in the numerator to ensure that the weight of an edge is never equal to zero.

% This choice of weights is inspired by a long line of research by prior work, cf.~\cite{onnela2007structure,goyal2010learning,granovetter1973strength,cheng2013epidemic}, which has attempted to connect the influence and propagation level between two nodes with graph parameters.

Some of the prior algorithms that we discuss in Section~\ref{exp:sec} rely on a measure of distance between two nodes. Since the edge weights represent the strength of the relations, it is conventional to use their ``opposite'' form when calculating distance. More precisely, for an edge $(v,u)$, we use $1-\omega((v,u))$.

% where $\epsilon = \min_{(v,u)\in E}(\omega((v,u))*0.001)$ is added to avoid zero values.
% when computing distance, where $\epsilon>0$ is a very . Here $\epsilon$ is defined as $\epsilon = {min}_{\{u,v\}\in E}(weights(u,v)*0.001)$.

\subsubsection{Algorithm Parameters.} For our algorithm, as discussed in Section~\ref{algorithm}, we need to set the initial resolution parameter $r$, the repetitions $h_1$ and $h_2$, and increasing factor $f > 1$. In our experiments, we set $r=0.01$ for Facebook and Facebook-Politician and $r=0.05$ for Facebook-Govt, $f=1.05$, and $h_1=h_2=5$. Note that the closer $f$ is to $1$ and the larger $h_1$ and $h_2$ are, the more precise our algorithm would be. There is nothing specifically unique about these choices. They are just some reasonable choices that allow our algorithm to perform well on the datasets used, as will be discussed in Section~\ref{exp:sec}.

\subsubsection{Containment Factor.} As mentioned, the Independent Cascade model serves as a simulation tool to emulate the process of misinformation spreading. Initially, a set $R$ of nodes is red and the rest is white. To measure the effectiveness of an edge blocking algorithm that blocks edges in a set $S$, we rely on \textit{containment factor} 

\begin{equation}
cf = 100\cdot\frac{\phi(G(V,E, \omega),R)-\phi(G(V,E\setminus S,\omega),R)}{\phi(G(V,E,\omega),R)}.
\end{equation}

Here $\phi(G(V,E,\omega),R)$ and $\phi(G(V,E\setminus S,\omega),R)$ denote the expected final number of orange nodes (when initially nodes in $R$ are red) before and after blocking edges in $S$. (Note that we focus on orange nodes, since all red nodes eventually become orange.) Thus, $\phi(G(V,E,\omega),R)$ is the number of nodes that become orange before blocking any edges, and $cf$ measures what percentage of them will remain white once edges in $S$ are blocked.

Note that maximizing $cf$ is the same as maximizing the final number of white nodes, used in the \textsc{Edge Blocking Problem}. To be consistent with prior work, cf.~\cite{zareie2022rumour}, we use $cf$ in our evaluations to compare the algorithms.

% spreading ability (number of active nodes at the end of spreading process) before and after removing of edges respectively. At the end of the spreading process, the total number of active nodes show the  spreading ability of the misinformation, denoted by $\phi$. 

\subsection{Comparison of Algorithms}
\label{exp:sec}
We compare our proposed algorithm against algorithms from prior work. 
\begin{itemize}
    \item \textbf{RNDM:} A set of edges is randomly selected to be blocked.
    \item \textbf{HWT:} Edges with the largest weight are blocked.
    \item \textbf{DEG~\cite{kempe2003maximizing,yan2019rumor}:} The edges for which the sum of the degree of their two endpoints are the largest are blocked.
    \item \textbf{WDEG:} This is the same as DEG, except the weighted degrees (the sum of the weight of adjacent edges for each node) are considered.

    % The edges with highest sum of  weight of edges connected to nodes $v_i$ and $v_j$ of edge $e_{ij}$ are blocked.
    \item \textbf{CLO:} The edges for which the sum of the closeness of their two endpoints are the largest are blocked.
    \item \textbf{WCLO:} This is the same as CLO, except the edge weights (their ``opposite`` actually, as explained in Section~\ref{setup:sec}) are considered when calculating closeness.
    \item \textbf{BET~\cite{dey2017centrality}:} The edges with the highest betweenness centrality are blocked.
    \item \textbf{WBET:} The edges with the highest weighted betweenness are blocked.
    \item \textbf{PGRK~\cite{brin1998anatomy,yan2019rumor}:} The edges for which the sum of the PageRank centrality of their two endpoints are the largest are blocked.
    % The sum of the pagerank centrality of the endpoints of each edge can determine the importance of the edge in the network. It is considered as a heuristic in \cite{} to identify the edges to be blocked.
    \item \textbf{IEED~\cite{zareie2022rumour}:} In each iteration, a ``critical'' edge is determined and blocked from the network. Criticality is determined using nodes’ influence and edges’ blocking efficiency, weighed using a notion of entropy. (Please refer to~\cite{zareie2022rumour} for more details on this algorithm.)
    
    % An edge blocking method is used that considers different features of edges to identify a set of critical edges whose blocking minimizes the spread of misinformation. A critical edge is determined iteratively and blocked from the network. Criticality is determined using nodes’ influence and edges’ blocking efficiency.
\end{itemize}

For each of our three networks, we select a randomly chosen set $R$ of nodes of size $|R|= 0.001n$ to be red initially (and the rest white). We let the number of blocked edges to range from $0.01 m$ to $0.2 m$. Then, we compute the containment factor $cf$ for all the algorithms by blocking the corresponding edges and running the Independent Cascade model. For each experiment, we select $\mid R\mid$ nodes to be red, and then run the Independent Cascade Model 10 times to obtain the $cf$ for the same set of initial red nodes. We run each of these experiments 10 times for different sets of initial red nodes and report the average value of $cf$. (The standard deviations are given in Appendix~\ref{appendix:sd}.) 

% The number of sources of misinformation is fixed at $|R|= 0.001n$, while the number of blocked edges, (represented by $\mid$S$\mid$) varies from $\mid$S$\mid = 0.01\cdot \mid $E$\mid$ to $\mid$S$\mid = 0.2\cdot \mid $E$\mid$.
% 10 times.

% All the methods above and the proposed algorithms were implemented in Python, and libraries, Numpy and Networkx were used.

\begin{figure}[htp]
    \centering
    \begin{minipage}[c]{0.45\linewidth}
        \includegraphics[width=6.5 cm]{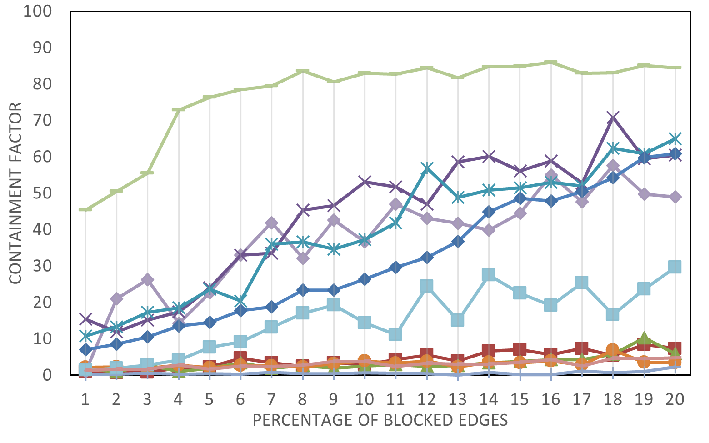} 
        % \caption{Facebook}
    \end{minipage}
    \begin{minipage}[c]{0.45\linewidth}
        \includegraphics[width=6.5 cm]{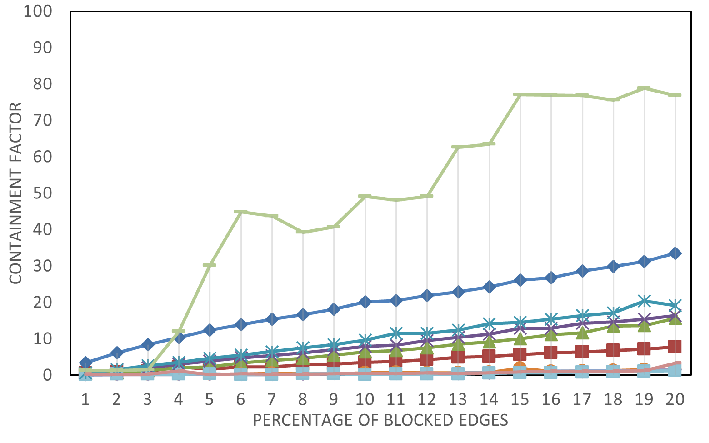}
        % \caption{FB-Govt}
    \end{minipage}
    \begin{minipage}[c]{0.45\linewidth}
        \includegraphics[width=6.5 cm]{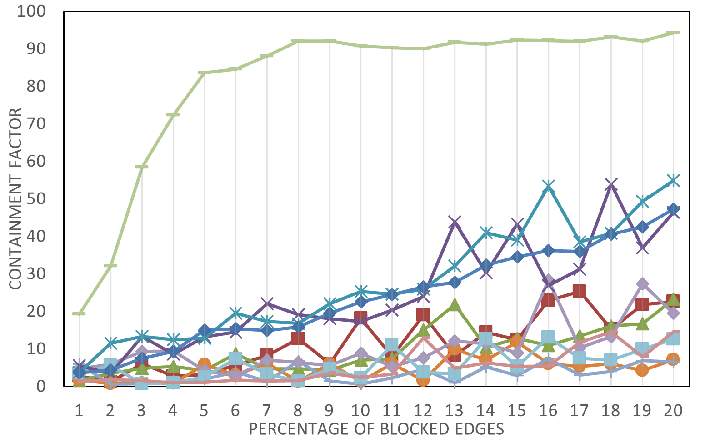}
        % \caption{FB-Politician}
    \end{minipage}
    \includegraphics[width=12cm]{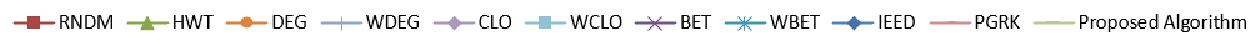}
    \label{fig: results}
    \caption{The containment factor for different algorithms on Facebook (top-left), Facebook-Govt (top-right), and Facebook-Politician (bottom) networks.}
\end{figure}

The outcomes of our experiments are provided in Fig.~2. The vertical axis denotes the containment factor of the algorithms, while the horizontal axis is the percentage of edges blocked.
As expected, it can be seen that as the percentage of blocked edges increases, the containment factor of the methods increases. We observe that our proposed algorithm consistently outperforms all other algorithms, especially by a significant margin for higher percentages of blocked edges. Our proposed algorithm is followed by BET, WBET, and IEED. The only case where our algorithm does not perform better than the other algorithms is for small percentages of blocked edges on the Facebook-Govt dataset.
% This is because our algorithm obtains a low number of communities for this range. Hence, most of the nodes belong to the same communities, and disconnecting the communities does not have much effect because the misinformation can potentially still affect a large number of nodes. 
% For other values of s, our algorithms can obtain many communities and outperform other algorithms.

\section{Conclusion} \label{conclusion}
We studied the problem of mitigating misinformation spreading in social networks using blocking edges. After providing a formal formulation of the problem, we proved that it is NP-hard. Then, we proposed an intuitive community-based algorithm, which first partitions the node set into well-connected communities by leveraging the Louvain algorithm. Then, it blocks the inter-community edges. Through experiments on real-world social networks, we observed that this algorithm, despite its simplicity, consistently and significantly outperforms the prior algorithms.
% We proposed a new algorithm for misinformation mitigation in social networks that uses community detection to identify the edges such that removing them contains the misinformation within their community of origin. Our algorithm can be highly beneficial in reducing the spread of misinformation on large, well-connected social networks, like Facebook, Instagram, WhatsApp, etc. 
% By deploying our algorithm in large social networks, such as facebook, twitter, Instagram etc, we can significantly reduce the impact of misinformation, protect users from misinformation, and promote a more trustworthy and reliable online environment. 

There are several potential future research avenues. It would be interesting to devise more effective strategies to choose the final resolution parameter in our proposed algorithm such that the number of inter-community edges is as close as possible to the budget $k$. Furthermore, other community detection algorithms can be explored, rather than the Louvain algorithm. One also might apply a community detection algorithm to devise a \textit{node} blocking strategy. Finally, studying the \textsc{Edge Blocking Problem} where each edge has a given cost would be interesting from both a practical and theoretical perspective.

% Our paper can be extended in different ways. Defining a better way to find a suitable Resolution parameter such that $\mid S\mid$ is close to the desired s may improve the accuracy of the proposed method. Other community detection algorithms techniques can be tried. Using community detection to block nodes or assign misinformation verifiers can be studied as well. 

\bibliographystyle{splncs04}
\bibliography{ref}

\appendix

\section{Standard Deviations}
\label{appendix:sd}
The standard deviation obtained for the proposed algorithm in $cf$ for the three datasets is given in Table \ref{standard deviation}. 
\begin{table}[]
    \centering
    \caption{Standard deviation obtained for the proposed algorithm.}
    \label{standard deviation}
    \begin{tabular}{| c | c | c | c |}
        \hline
        Percentage & Facebook & Facebook-Govt & Facebook-Politician \\  \hline
        \textbf{1} & 13.13838739 & 0.536237924 & 13.62691459 \\
        \textbf{2} & 11.47495844 & 0.41963079 & 11.68562978  \\ 
        \textbf{3} & 7.922823641 & 0.406590157 & 12.87673285  \\ 
        \textbf{4} & 7.828210311 & 3.765703269 & 11.54790164  \\ 
        \textbf{5} & 5.521524447 & 12.09778772 & 6.344192445  \\ 
        \textbf{6} & 5.794904371 & 9.049993554 & 8.917378102  \\ 
        \textbf{7} & 4.393543369 & 10.34596116 & 6.689110222  \\ 
        \textbf{8} & 6.45707192 & 9.882091322 & 4.062169098  \\ 
        \textbf{9} & 7.01420812 & 9.639261152 & 4.078979175  \\ 
        \textbf{10} & 4.586620155 & 11.71557316 & 4.285124269  \\ 
        \textbf{11} & 5.191117306 & 9.887803037 & 6.042205631  \\ 
        \textbf{12} & 4.300630574 & 9.524686288 & 3.913324957  \\ 
        \textbf{13} & 3.196539101 & 9.898218078 & 2.538393193  \\ 
        \textbf{14} & 4.516281164 & 9.493850407 & 3.342031983  \\ 
        \textbf{15} & 5.215616401 & 10.76495219 & 3.927128581  \\ 
        \textbf{16} & 5.813759159 & 4.749856723 & 4.339572944  \\ 
        \textbf{17} & 2.844741933 & 5.597944464 & 2.68540086  \\ 
        \textbf{18} & 5.075444151 & 7.42852049 & 2.974853124  \\ 
        \textbf{19} & 3.379384559 & 5.622314272 & 3.102571693  \\ 
        \textbf{20} & 3.678885097 & 5.077641732 & 2.681650404  \\ \hline
    \end{tabular}
\end{table}

\end{document}